\documentclass[journal]{IEEEtran}

\usepackage{amsmath,amssymb,amsfonts}
\usepackage{graphicx}
\usepackage{booktabs}
\usepackage{array}
\usepackage{url}
\usepackage{xcolor}
\usepackage[hidelinks]{hyperref}
\usepackage{cite}

\graphicspath{{figures/}}

\newtheorem{proposition}{Proposition}
\newtheorem{corollary}{Corollary}
\newtheorem{assumption}{Assumption}

\newcommand{\E}{\mathbb{E}}
\newcommand{\Var}{\mathrm{Var}}

\begin{document}

\title{Width, Memory, and Delay: A Resource Accounting\\ for the Limits of Flat Multi-Agent Systems}

\author{Oleksandr~Kuznetsov and Emanuele~Frontoni
\thanks{Manuscript prepared \today.}
\thanks{O.~Kuznetsov is with the Department of Theoretical and Applied Sciences,
eCampus University, Via Isimbardi 10, 22060 Novedrate (CO), Italy; with the
SMARTEST Research Center, eCampus University, Via Isimbardi 10, 22060 Novedrate
(CO), Italy; and with the Department of Intelligent Software Systems and
Technologies, School of Computer Science and Artificial Intelligence,
V.~N.~Karazin Kharkiv National University, 4 Svobody Sq., 61022 Kharkiv,
Ukraine (e-mail: oleksandr.kuznetsov@uniecampus.it; kuznetsov@karazin.ua;
ORCID: 0000-0003-2331-6326).}
\thanks{E.~Frontoni is with the Department of Political Sciences, Communication
and International Relations, University of Macerata, Via Crescimbeni, 30/32,
62100 Macerata, Italy (e-mail: emanuele.frontoni@unimc.it; ORCID:
0000-0002-8893-9244).}}

\markboth{IEEE Transactions on Cybernetics}%
{Kuznetsov \MakeLowercase{\textit{et al.}}: Width, Memory, and Delay in Flat Multi-Agent Systems}

\maketitle

\begin{abstract}
A recurring question in the design of scalable multi-agent systems---from robot
swarms to collectives of large-language-model (LLM) agents---is whether adding
more agents can, on its own, overcome performance limits, or whether a
qualitatively \emph{deeper} organization is required. A recent preprint argues
that flat, homogeneous multi-agent systems face an irreducible,
population-independent ``causal floor'' on achievable error, removable only by
hierarchical (nested-loop) organization. Using a controlled disturbance-rejection
testbed with an exactly computable optimum, we show this conclusion is too
strong and replace it with a quantitative resource model built on three
resources: population \emph{width} $N$, per-agent internal-model \emph{memory}
$d$, and prediction across the observation \emph{delay} $\tau$. We establish
three claims. (i) The achievable floor is governed not by architectural
hierarchy but by per-agent internal-model content: a flat, homogeneous swarm
whose agents carry a matched internal model of the disturbance matches or beats
a designed two-loop hierarchy at equal per-agent memory---so temporal depth can
be dynamical (recurrent memory), not architectural (nesting). (ii) The three
resources are \emph{not mutually interchangeable}; we chart the exchange rates
and the hard non-exchange boundaries on an explicit width$\times$memory map,
including a strict equal-total-state-budget comparison. (iii) A residual floor
is set by the observation delay and the environment's unpredictability over that
horizon, which we verify against the optimal controller (the exact
minimum-variance characterization is given in Proposition~\ref{prop:delay}). We
quantify the price of replacing oracle knowledge of the disturbance spectrum with
online learning, provide a preliminary robustness check against a mild bounded nonlinearity and a
spatially-extended plant, and distill four design rules for practitioners. All
experiments are released with a reproducible pipeline.
\end{abstract}


\section{Introduction}
\IEEEPARstart{A}{practical} question now dominates the engineering of
multi-agent systems: does adding more agents keep paying off? In the
language-model-agent literature the debate is explicit and unresolved---several
recent studies report that a well-tuned single agent can match or beat elaborate
multi-agent workflows \cite{xu2026single,tran2026single}, while others find that
scaling helps only when the added agents bring genuine \emph{diversity} rather
than replication \cite{qian2024scaling}. The same tension appears in robot
swarms and networked control, where consensus and averaging arguments promise
gains with scale \cite{olfati2007,jadbabaie2003,ren2005}, yet some limits appear
stubbornly independent of population size.

A recent preprint \cite{kovalov_preprint} sharpens the pessimistic side into an
impossibility claim: flat, homogeneous multi-agent systems possess a
population-independent ``causal floor'' on achievable error, and only
hierarchical (nested-loop) organization can lower it. The intuition is
attractive---some limits do seem to demand a deeper temporal organization---but,
we argue, the claim conflates two distinct notions under the word ``flat'':
\emph{flat in architecture} (no nested control loops) and \emph{flat in memory}
(agents without internal dynamical state). The impossibility result holds for
the second notion and fails for the first, and the two come apart exactly where
the interesting engineering lives.

This paper replaces the single flat/deep axis with a \emph{resource triangle}
---width, memory, and delay-prediction (Fig.~\ref{fig:concept})---and asks, for each pair of resources,
whether one can be traded for another and at what rate. The answer is a mix of
clean exchange laws and hard prohibitions, all recoverable from classical
stochastic control theory (minimum-variance control \cite{astrom1970,astrom1973},
the internal model principle \cite{francis1976}, and the Bode/waterbed
constraints \cite{bode1945,seron1997}) once the problem is posed as resource
accounting rather than architecture. Our contribution is not a new control law;
it is the reframing, an explicit map of where exchange is and is not possible,
and a careful empirical verification---against a computable optimum---that pins
the constants. Because the accounting is stated in terms of averageable noise,
structured disturbance, and predictability across a delay, it is agnostic to
whether an ``agent'' is a controller, a robot, or an LLM; we return to this
bridge in the conclusion.

\begin{figure*}[!t]
\centering
\includegraphics[width=0.8\textwidth]{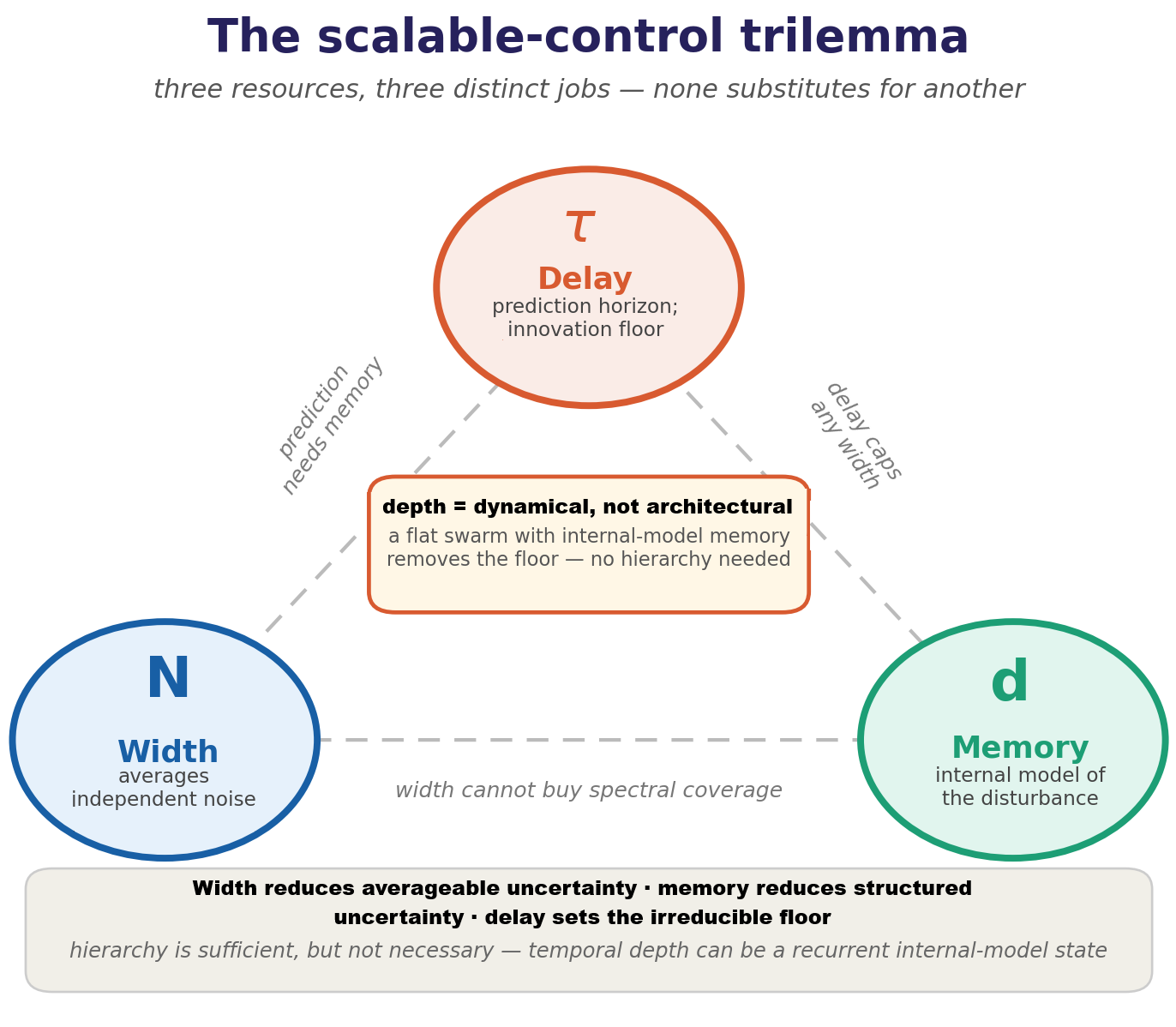}
\caption{The resource triangle. Width averages independent measurement noise;
per-agent internal-model memory covers the structured disturbance; delay sets an
irreducible innovation floor. A memoryless flat swarm cannot lower the
structured floor, but a flat swarm whose agents carry an internal model can---so
temporal depth may be dynamical (memory), not architectural (nesting).}
\label{fig:concept}
\end{figure*}

\subsection*{Contributions}
\begin{enumerate}
\item A counterexample to the strong flat-systems theorem: a flat homogeneous
swarm with matched per-agent internal models beats a designed hierarchy
(Section~\ref{sec:exp}, Figs.~\ref{fig:scaling} and~\ref{fig:map}).
\item A quantitative width$\times$memory map with exchange rates and
non-exchange boundaries, plus a fixed-budget ($N\cdot d$ const) analysis
(Figs.~\ref{fig:map} and~\ref{fig:budget}).
\item A delay theorem, stated in resource terms and verified to $<1\%$ against
the optimal predictive controller (Fig.~\ref{fig:delay}).
\item The cost of adaptation: how an oracle-free learned internal model
approaches the oracle, and where it fails (Figs.~\ref{fig:mismatch}
and~\ref{fig:nonstat}).
\item Four practitioner rules and a nonlinear-plant robustness check
(Section~\ref{sec:triangle}, Fig.~\ref{fig:nonlinear}), plus a
spatially-extended plant confirming the accounting is not an artifact of the
scalar testbed (Fig.~\ref{fig:spatial}).
\end{enumerate}

\section{Related Work}
\textbf{Minimum-variance and self-tuning control.} The floor we characterize is
the minimum-variance benchmark of {\AA}str{\"o}m \cite{astrom1970}: for a linear
plant with delay and a disturbance with rational spectrum, the minimum
achievable output variance equals the variance of the multi-step-ahead
prediction error of the disturbance, obtained by polynomial division of the
noise model. The self-tuning regulator \cite{astrom1973} identifies that model
online; Harris \cite{harris1989} turned the benchmark into a closed-loop
performance monitor still used industrially. Our adaptive baseline is in this
lineage in spirit---it identifies the disturbance online and cancels it---but is
not a full AR/RLS minimum-variance self-tuning regulator; it is an adaptive
internal-model controller specialized to narrowband disturbances (periodogram-
based frequency identification driving a bank of learned resonators,
Section~\ref{sec:setup}). We are explicit about this to avoid claiming more than
the implementation delivers.

\textbf{Internal model principle.} Francis and Wonham \cite{francis1976} showed
that asymptotic rejection of a disturbance requires the controller to embed a
model of that disturbance's generator. This is exactly the ``per-agent memory
buys spectral coverage'' axis: a resonator is the internal model of a sinusoid,
an integrator that of a constant bias, an AR model that of a colored process.

\textbf{Fundamental limitations.} Bode's sensitivity integral and its
delay-constrained descendants---the ``waterbed effect'' \cite{bode1945,seron1997}
---formalize that suppression pushed down in one band re-emerges elsewhere and
that delay caps the bandwidth over which suppression is possible. We use these to
explain the notch-depth/robustness-width trade-off (Fig.~\ref{fig:mismatch}) and
the delay wall (Fig.~\ref{fig:delay}) rather than re-deriving them. The Smith
predictor \cite{smith1957} and Kalman filter \cite{kalman1960} supply the optimal
predictive machinery for the delay analysis.

\textbf{Multi-agent scaling and emergence.} Consensus and cooperative control
establish $\propto 1/N$ averaging of independent measurement noise
\cite{olfati2007,jadbabaie2003,ren2005}. The LLM-agent literature is now testing
whether analogous scaling holds for cognitive agents, with mixed conclusions
\cite{xu2026single,qian2024scaling,tran2026single,cemri2025fail}. The target preprint
\cite{kovalov_preprint} sits at this intersection; we treat its claim as the
null hypothesis to be refined.

\textbf{Distributed estimation and graph-theoretic internal models.} We
deliberately study two extremes: fully decentralized agents with no
communication (our width axis) and a fully centralized Kalman fusion (our
frontier benchmark, Section~\ref{sec:setup}). Between them lies a large
literature this paper does not engage experimentally: distributed/consensus
Kalman filtering lets agents partially fuse measurements over a communication
graph without a fusion center \cite{olfati2005distributed}, and the internal
model principle has been extended to networked, graph-constrained settings via
cooperative output regulation, where each agent embeds a copy of the exogenous
signal generator and consensus over the graph substitutes for centralized
coordination \cite{su2012cooperative}. Limited communication could in principle
let width partially substitute for memory by letting agents share spectral
estimates; our resource accounting does not yet include this regime, and we
flag it as the most direct extension of this work (Section~\ref{sec:limits}).

\section{Problem Setup and Testbed}
\label{sec:setup}
From here through Section~\ref{sec:triangle} the setting is a physical control
system; the disturbance is a stochastic signal with known statistics, not a
model of cognition. The bridge to LLM-agent collectives---where ``memory'' is
context and retrieval and ``delay'' is reasoning latency---is a motivating
analogy, taken up only in the introduction, the discussion
(Section~\ref{sec:discussion}), and Appendix~\ref{app:explore}; it is not part
of the control-theoretic claims.

\textbf{Plant.} Our testbed is deliberately minimal: a \emph{single controlled
scalar object} observed by $N$ homogeneous agents whose control actions are
averaged and whose measurement noises are independent. Concretely, one state
$x$ evolves under a common disturbance and each agent $i$ contributes a control
$u_i$ from its own noisy, delayed observation:
\begin{equation}
x[t+1] = x[t] + \bar{u}[t] + d[t] \;(+\,\gamma \sin x[t]),\quad
\bar{u}[t] = \tfrac1N\!\sum_i u_i[t],
\end{equation}
where the optional term introduces plant nonlinearity. This is not a swarm of
$N$ independent plants but the canonical setting in which \emph{width} (more
agents) and \emph{memory} (per-agent internal state) can be varied
orthogonally; Section~\ref{sec:limits} discusses the spatially-extended
generalization. The common disturbance is
\begin{equation}
d[t] = \sum_b A_b \sin(\phi_b[t]) + p[t], \quad p[t] = a\,p[t-1] + e[t],
\end{equation}
with $\phi_b$ accumulating at frequency $\omega_b$ (permitting chirp and regime
switching) and $e[t]$ white. Each agent observes
\begin{equation}
y_i[t] = x[t-\tau] + v_i[t], \quad v_i \sim \mathcal{N}(0,\sigma_m^2),
\end{equation}
with $v_i$ \emph{independent across agents}---this is what population width can
average away; the common disturbance is what it cannot. \emph{The delay is in
the observation, not in the actuation}: $\bar{u}[t]$ enters $x[t+1]$
immediately, but each agent must choose $u_i[t]$ from information no fresher
than $x[t-\tau]$. The control problem is therefore predictive---the controller
must forecast the state and disturbance over a horizon of $\tau+1$ steps from
the delayed, noisy observation.

\textbf{Aggregate view for centralized baselines.} For the mean of $N$ agents
with i.i.d.\ measurement noise, $\bar{y}[t] = x[t-\tau] + v[t]/\sqrt{N}$, so
width enters the strong baselines \emph{only} as measurement-noise reduction---a
clean separation of the width axis from the memory axis. The Kalman frontier is
a \emph{centralized} reference built on this ensemble mean: it is the best
achievable performance \emph{if the agents' measurements could be perfectly
fused}. Decentralized agents with no communication cannot attain it---each sees
only its own noisy $y_i$---so we use it purely as an upper-bound benchmark and a
computable optimum for the resource accounting, not as a controller the swarm
could implement.

\textbf{Metric.} Steady-state MSE $=\E[x^2]$ after burn-in, averaged over seeds
with $95\%$ confidence intervals. We also fit $\mathrm{MSE}(N)=A N^{-\alpha}+C$
and report the floor $C$ with bootstrap CIs.

\textbf{Controllers.}
\begin{itemize}
\item \emph{Ladder (per-agent memory $d$)}: a homogeneous flat swarm whose
agents carry an internal model of increasing order---P ($d{=}0$), PI ($d{=}1$),
and PI+resonator banks (PIRES, $d = 1 + 2k$ for $k$ resonators). The integrator
covers bias and very-low-frequency drift; each resonator adds one narrowband
component, so memory depth is natural and monotone.
\item \emph{Cascade}: a two-loop hierarchy (the preprint's ``deep'' system).
\item \emph{Kalman frontier}: oracle-model Kalman-predictive minimum-variance
control; the computable optimum, with a semi-analytic Riccati floor.
\item \emph{Adaptive}: oracle-free adaptive internal-model control---identifies
the disturbance tones online (peak-picking on an exponentially-averaged
periodogram of a decimated, pre-filtered reconstruction) and cancels them with
phase-locked-loop learned resonators plus soft feedback. It is a \emph{heuristic engineering baseline}, not a self-tuning
minimum-variance regulator: it uses periodogram peak-picking rather than
recursive least squares and has no formal convergence guarantee. We include it
to quantify the gap between oracle frequency knowledge and practical online
learning. It also carries additional estimator memory (a spectral buffer,
oscillator bank, and filter states) beyond $2k{+}1$ states, so we report it as a
no-oracle baseline and exclude it from equal-memory comparisons.
\end{itemize}

\section{Framework: Three Resources}
We define the resources so the exchange questions are well posed.

\textbf{Width $N$.} The number of agents sharing the common disturbance and
averaging independent measurement noise; the resource that reduces the
\emph{averageable} error component.

\textbf{Per-agent memory $d$.} The dimension of each agent's internal dynamical
state, updated as $m_i[t+1] = f(m_i[t], y_i[t])$. We distinguish
\emph{internal-model} memory (state encoding a generator of the disturbance:
integrator, resonator, AR model) from generic recurrence; only internal-model
memory buys spectral coverage.

\textbf{Temporal depth / prediction.} The organization of memory across
separated timescales, used to predict the plant state and disturbance over the
delay $\tau+1$ from the delayed, noisy observation. This is what the preprint
intuits as ``depth''; we argue it is not architectural nesting but the presence,
in the loop, of a recurrent internal model whose prediction horizon spans $\tau$.
When that model is fixed we call it \emph{recurrent internal-model memory}; when
it also updates its model online and uses it to predict across the delay we call
it \emph{adaptive predictive memory}. (We avoid ``recursive'': in a control/RNN
context ``recurrent'' is accurate.)

The refined thesis, in one line: \emph{Population width reduces only averageable
uncertainty. Internal-model memory reduces structured uncertainty. Temporal
depth is not architectural hierarchy but predictive closed-loop state organized
across the relevant timescales. Delay sets the irreducible innovation floor.}

\section{Propositions}
We state the results as propositions with proof sketches; empirical verification
follows in Section~\ref{sec:exp}.

\begin{assumption}\label{as:all}
(A1) linear plant; (A2) disturbance with rational spectrum (finite sum of
narrowband tones plus an AR process); (A3) per-agent i.i.d.\ measurement noise,
independent across agents; (A4) known plant delay $\tau$. Under
(A1)--(A4) the minimum-variance machinery of \cite{astrom1970} applies to the
aggregate.
\end{assumption}

\begin{proposition}[Width--structure non-exchange]\label{prop:width}
Consider controllers of a fixed class and bounded bandwidth, and let the
disturbance contain power $P_b$ in a frequency band whose generator is not
embedded in any agent's internal model. Then the steady-state error contributed
by that band is bounded below by a positive constant independent of $N$;
increasing $N$ cannot reduce it. (Width lowers the \emph{averageable} noise
component but not the structured residual, whose \emph{asymptotic} rejection
requires an internal model of the band; a broadband robust controller may
partially attenuate the band but cannot drive its contribution to zero.)
\end{proposition}
\begin{IEEEproof}[Proof sketch]
Width acts only on the per-agent i.i.d.\ term $v_i/\sqrt{N}$ in the aggregate;
the common disturbance $d[t]$ is identical across agents and survives averaging
unchanged. Rejecting power in band $b$ requires loop gain shaped at $\omega_b$,
which by the internal model principle \cite{francis1976} requires an internal
model of that band in the loop. Absent such a model, the closed-loop sensitivity
at $\omega_b$ is $\Theta(1)$ in $N$, so the band's contribution to $\E[x^2]$ is
bounded below by a constant. Verified in Fig.~\ref{fig:bands} and by the
L-shaped contours of Fig.~\ref{fig:map}.
\end{IEEEproof}

\begin{proposition}[Memory--delay partial exchange; the delay floor]\label{prop:delay}
Under Assumption~\ref{as:all}, the minimum achievable steady-state variance
equals the conditional variance of the controlled next state given the delayed,
noisy observation filtration:
\begin{equation}
C_{\min}(\tau) = \Var\!\big(x[t+1] - \hat{x}[t+1 \mid \mathcal{F}_t^y]\big),
\end{equation}
$\mathcal{F}_t^y = \sigma(y[s], u[s{-}1] : s \le t)$, $y[s]=x[s-\tau]+v[s]$. This
has two contributions: (a) the innovation of the disturbance accumulated over
the horizon $\tau+1$, which no controller can anticipate, and (b) the residual
state-estimation uncertainty from observing only the delayed, noise-corrupted
$x[t-\tau]$. Under perfect state measurement it reduces to the classical
minimum-variance floor---the variance of the $(\tau{+}1)$-step disturbance
prediction error \cite{astrom1970}---but with noisy delayed sensing the
estimation term must be included.
\end{proposition}
\begin{IEEEproof}[Proof sketch]
Because control acts through the plant without its own delay, $x[t+1]$ is fixed
by $x[t]$, the applied $u[t]$, and $d[t]$; the best mean-square $u[t]$ cancels
$\E[x[t]+d[t]\mid\mathcal{F}_t^y]$, leaving $\Var(x[t+1]\mid\mathcal{F}_t^y)$.
Writing $d=(C/A)e$ with $e$ white and applying the $(\tau{+}1)$-step division
$C/A = F + z^{-(\tau+1)}G/A$ isolates the unpredictable innovation $Fe$; the
delayed noisy observation adds the Kalman estimation covariance of the current
state, which the Riccati recursion of the optimal filter \cite{kalman1960}
yields in closed form. The sum is $C_{\min}(\tau)$; it grows with $\tau$, with
the disturbance memory (larger $f_k$ for more strongly colored spectra), and
with the sensing noise $\sigma_m^2/N$. Verified to $<1\%$ in
Fig.~\ref{fig:delay}.
\end{IEEEproof}

\begin{corollary}[Delay wall]
No amount of width or memory removes $\Var(Fe)$; it is the environment's
unpredictability over the latency horizon.
\end{corollary}

\begin{proposition}[Depth without hierarchy]\label{prop:depth}
There exists a flat, homogeneous swarm (identical agents, no nested loops) whose
per-agent internal model matches the disturbance and whose steady-state floor is
at or below that of a designed two-loop hierarchy with the same per-agent
memory. Hence architectural hierarchy is sufficient but not necessary to lower
the floor.
\end{proposition}
\begin{IEEEproof}[Proof sketch]
Constructive. Give each agent a PIRES internal model tuned to the disturbance
bands (memory $d = 1 + 2\cdot\#\text{bands}$). With the bands cancelled, the flat
swarm's floor is the delay/innovation floor of Proposition~\ref{prop:delay}. It
suffices to exhibit one flat controller that beats one designed hierarchy: our
PIRES swarm attains a lower floor than the implemented two-loop cascade at equal
per-agent memory (Fig.~\ref{fig:scaling}), so architectural hierarchy is not
necessary in general. (We do not claim the flat controller beats \emph{every}
hierarchy---a better-tuned or differently-organized cascade might match it; the
impossibility claim we refute is universal, so a single counterexample is
enough.) Verified in Fig.~\ref{fig:scaling} and the $d{=}7$ row of
Fig.~\ref{fig:map}.
\end{IEEEproof}

\begin{proposition}[Cost of adaptation, empirical]\label{prop:adapt}
The heuristic no-oracle controller of Section~\ref{sec:setup} approaches a tuned
internal-model (oracle-frequency) baseline in stationary narrowband settings,
with a gap that is controlled empirically by the online identification error and
shrinks as the per-population SNR (hence $N$) grows. Below an identification
threshold $N_{\mathrm{id}}$ its floor is uncontrolled, and under regime switches
faster than the identification time it is dominated by re-identification
transients and may exceed that of a robust memoryless loop. We state this as an
empirical observation, not a convergence theorem: our controller is an
engineering baseline, not a self-tuning minimum-variance regulator, and a formal
convergence analysis is beyond this paper's scope. A tighter regulator (e.g.\
RLS) would likely narrow the stationary gap; the transient failure under fast
switching, however, is set by the re-identification time and is therefore
fundamental to any learning scheme, not an artifact of our estimator.
\end{proposition}
\begin{IEEEproof}[Justification]
For the classical AR/RLS self-tuning regulator the attained limit is the
minimum-variance floor \cite{astrom1973}; our resonator-bank controller only
shares the same qualitative behavior for narrowband disturbances, without a
convergence guarantee. Identifiability requires the disturbance to rise above the
reconstruction noise, which scales as $\sigma_m^2/N$; below that the periodogram
peaks are buried. Under switching faster than the estimator's memory the estimate
never settles. Observed in the adaptive row of Fig.~\ref{fig:map} and in
Figs.~\ref{fig:mismatch}--\ref{fig:nonstat}.
\end{IEEEproof}

\section{Experiments}
\label{sec:exp}
Phase~1 (E1--E5) establishes the dichotomy and the counterexample; Phase~2
(Blocks A--E) turns it into the quantitative accounting and adds the optimal
frontier and the oracle-free adaptive baseline. Two exploratory studies (E6,
E7) are deferred to Appendix~\ref{app:explore}.

\textbf{E1--E3 --- counterexample (Fig.~\ref{fig:scaling}).} Floors ($95\%$ CI)
on a single-band task: P $1.40$, PI $0.32$, PID $0.12$, CASCADE $0.107$, flat
RES $0.058$. The flat resonator swarm ($d{=}2$) beats the two-loop cascade at
equal per-agent memory. To make the comparison fair to the hierarchy, the
cascade's loop gains were selected by a grid search (48 combinations over its
proportional, slow-loop-bandwidth, and slow-loop-integral gains) minimizing MSE
on the tuning task; the best-tuned gains are then frozen and reused across all
$N$, exactly as for every other controller. The flat swarm still attains the
lower floor. All scaling exponents $\alpha \approx 1.0$
(pure $1/N$ noise reduction to a fixed floor). \emph{Supports
Prop.~\ref{prop:depth}.}

\begin{figure}[!t]
\centering
\includegraphics[width=\columnwidth]{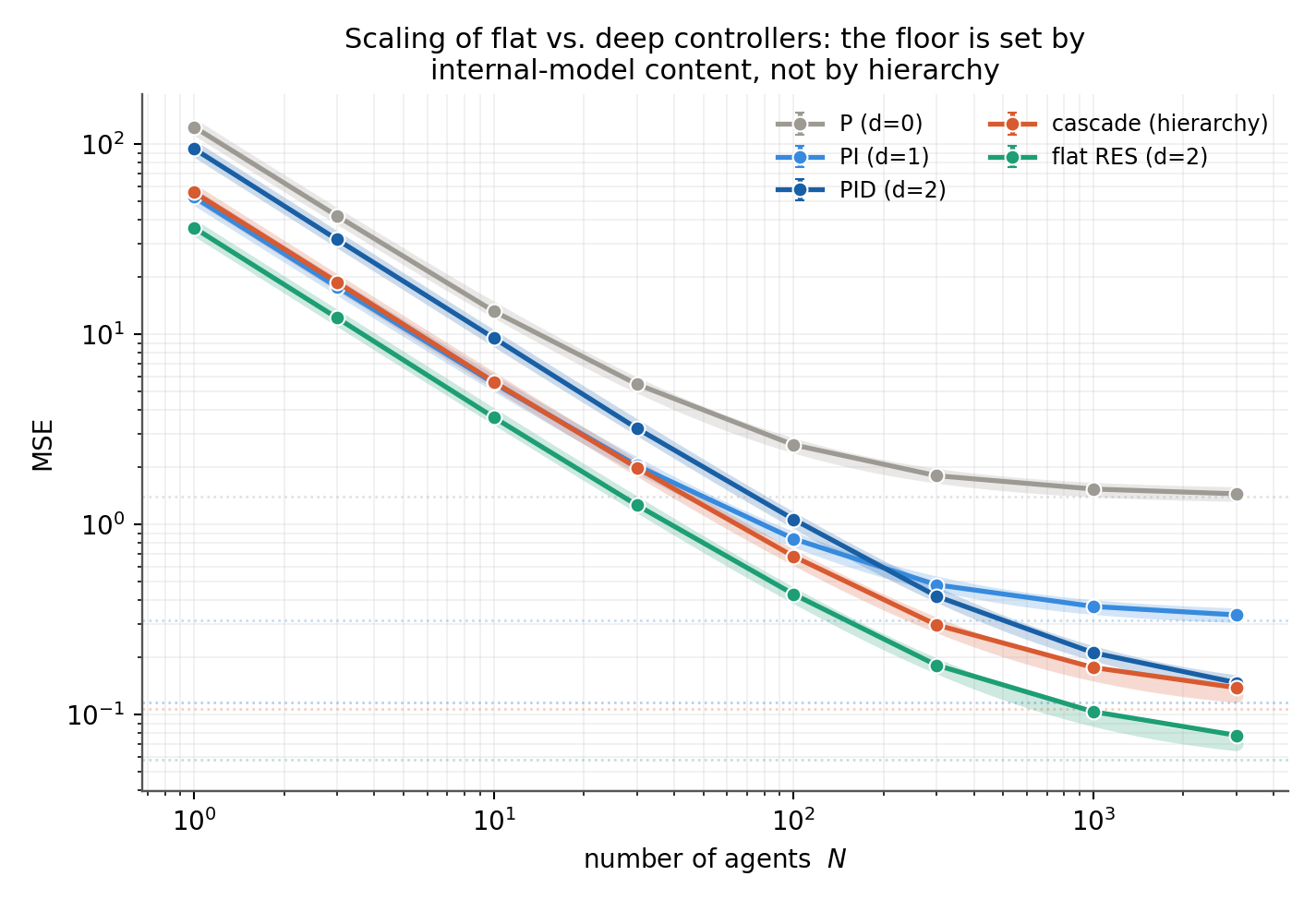}
\caption{E1--E3: MSE$(N)$ with fitted floors. A flat resonator swarm ($d{=}2$)
attains a lower floor than the two-loop cascade at equal per-agent memory,
refuting the strong flat-systems claim by counterexample.}
\label{fig:scaling}
\end{figure}

\textbf{E4 --- width cannot buy coverage (Fig.~\ref{fig:bands}).} On a two-band
task, an internal model covering one band leaves the other's power
$N$-independent (MSE at $N{=}1000$: $2.96$ uncovered vs.\ $0.25$ both covered).
\emph{Supports Prop.~\ref{prop:width}.}

\begin{figure*}[!t]
\centering
\includegraphics[width=0.98\textwidth]{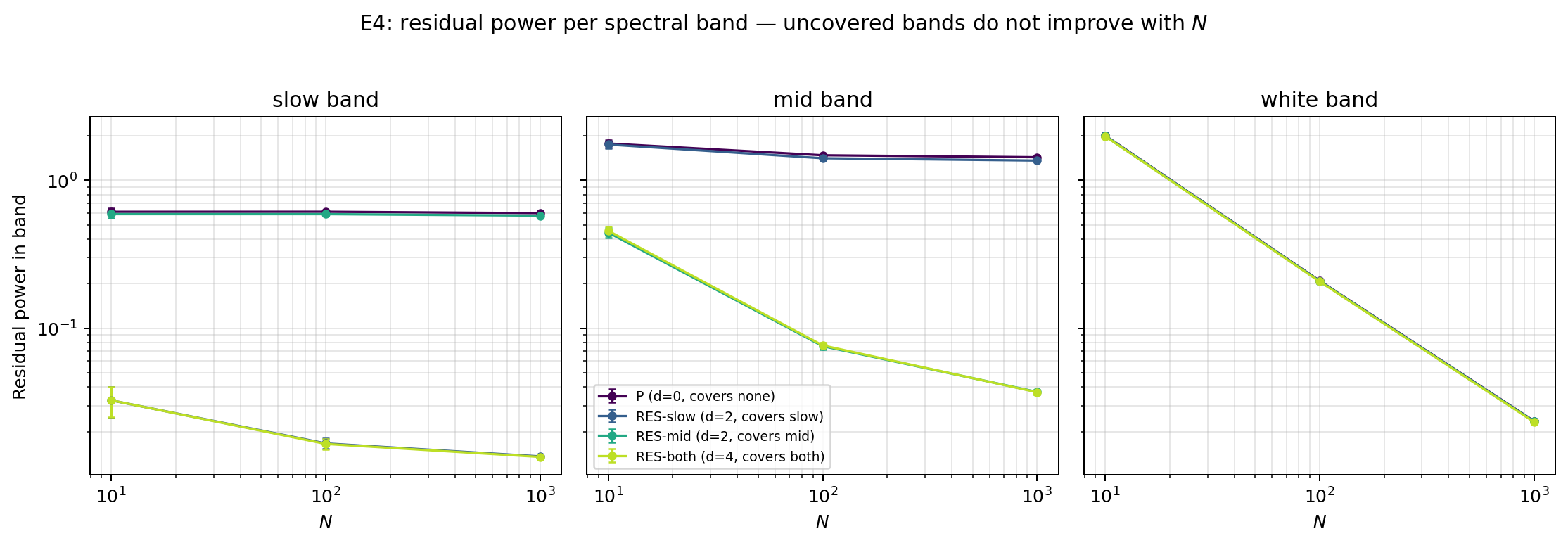}
\caption{E4: an internal model covering one band drives its power down with $N$;
the uncovered band's power is $N$-independent. Width cannot buy spectral
coverage.}
\label{fig:bands}
\end{figure*}

\textbf{E5 --- latency waterbed.} PI MSE rises $0.61\!\to\!3.29$ as
$\tau\!:\!1\!\to\!10$ and destabilizes by $\tau{=}20$; the PSD shows the
waterbed. \emph{Motivates Prop.~\ref{prop:delay}.}

\textbf{Block A --- the width$\times$memory map
(Figs.~\ref{fig:map},~\ref{fig:budget}).} MSE over $N\times d$ on a 3-band task.
Floors at $N{=}1000$: $d{=}0$ $4.49$, $d{=}1$ $2.40$, $d{=}3$ $2.05$, $d{=}5$
$1.66$, $d{=}7$ $0.44$ (monotone, since the PIRES ladder adds an integrator then
one resonator per band); no-oracle adaptive baseline $0.47$. The heatmap contours are L-shaped:
width slides you down a noise ramp to a fixed floor; memory lowers the floor
itself. A strict equal-budget comparison (exact $N\cdot d = B$ points,
Fig.~\ref{fig:budget}) shows a threshold: at the smallest budget ($B{=}1050$)
going deep is roughly a wash ($d{=}7$ with $N{=}150$ gives $2.24$ vs.\ $2.37$
for $d{=}1$ with $N{=}1050$), because too few agents leave the measurement noise
unaveraged; but past a minimum SNR ($B\ge 3150$) spending the same states on
memory dominates decisively ($d{=}7$: $0.80$ at $B{=}3150$, $0.29$ at
$B{=}12600$, vs.\ $\approx 2.3$ for $d{=}1$ at every budget). This directly
answers the ``you just gave each agent more memory'' objection.
\emph{Quantifies Props.~\ref{prop:width} and~\ref{prop:depth}.} The mechanism
behind the ladder is visible in the residual spectrum (Fig.~\ref{fig:psd}):
width scales the broadband noise floor, while each unit of internal-model
memory notches out one more disturbance band.

\begin{figure*}[!t]
\centering
\includegraphics[width=0.92\textwidth]{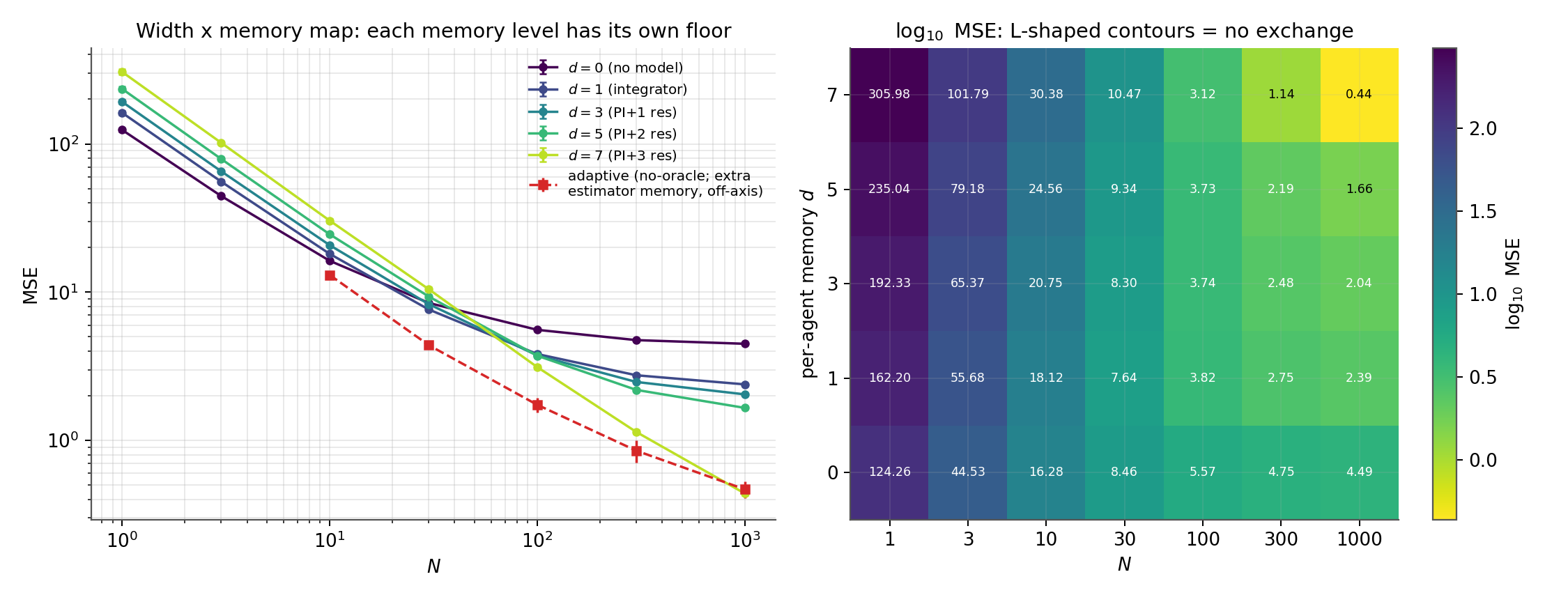}
\caption{Block A: MSE over $N\times d$ on a 3-band disturbance. Left: each
per-agent memory level has its own floor; the adaptive controller (red) starts
worse at small $N$ and approaches the oracle $d{=}7$ floor at large $N$. Right:
heatmap colored by $\log_{10}$ MSE (annotated numbers are the raw MSE, not its
log); the L-shaped contours mean width and memory are not
interchangeable---width slides down a noise ramp to a fixed floor, while only
memory (covering another band) lowers the floor itself.}
\label{fig:map}
\end{figure*}

\begin{figure}[!t]
\centering
\includegraphics[width=\columnwidth]{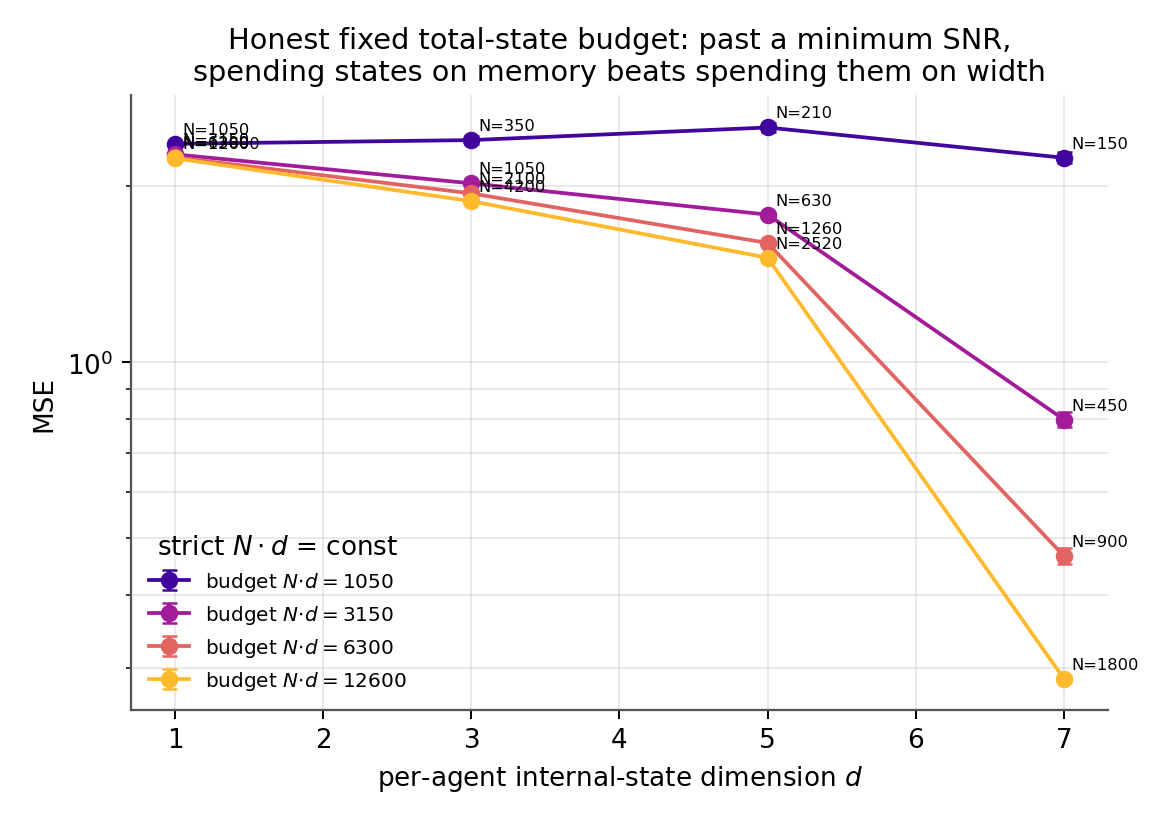}
\caption{Strict equal-budget comparison ($N\cdot d=B$, exact points, budgets
multiples of $\mathrm{lcm}(1,3,5,7)$). At the smallest budget ($B{=}1050$)
measurement noise dominates, so width and depth are equivalent; above a threshold
SNR, depth (internal-model memory) dominates. This answers the ``you just gave
each agent more memory'' objection.}
\label{fig:budget}
\end{figure}

\begin{figure}[!t]
\centering
\includegraphics[width=\columnwidth]{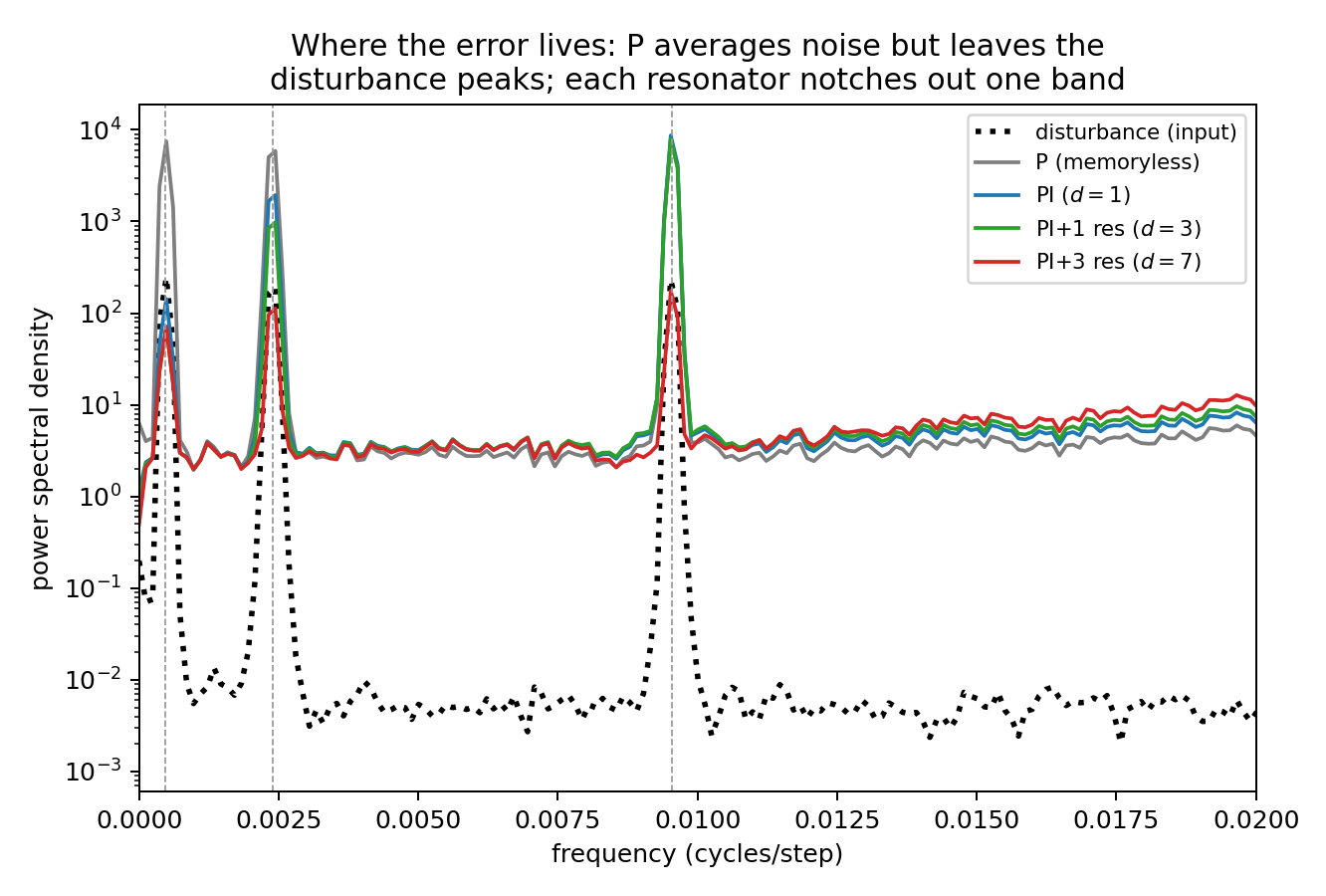}
\caption{Mechanism, in the spectrum. The memoryless P controller averages noise
but leaves the disturbance peaks intact; the integrator suppresses the lowest
band; each resonator notches out one more band. Memory reshapes the residual
spectrum, width only scales it---the internal-model principle made visible.}
\label{fig:psd}
\end{figure}

\textbf{Block B --- robustness and the oracle's price
(Figs.~\ref{fig:mismatch},~\ref{fig:robust2d},~\ref{fig:nonstat}).} \emph{Mismatch $\times$ damping}:
a sharp resonator ($\rho{=}0.9999$) reaches MSE $\approx 0.22$ within $\pm10\%$
of $\omega$ but degrades to $\approx 0.64$ at $30\%$ detune---worse than
model-free PI ($0.49$); a soft resonator ($\rho{=}0.99$) is $\approx 1.0$
everywhere. Depth of suppression $\leftrightarrow$ width of robustness is a
conserved trade-off (a Bode waterbed inside the notch), which the full
mismatch$\times$damping surface (Fig.~\ref{fig:robust2d}) confirms is continuous
across the whole parameter plane, not just at the three damping levels plotted
in Fig.~\ref{fig:mismatch}. The adaptive controller
($0.42$) beats the \emph{strongly} detuned oracles ($\ge 30\%$ error) and the
soft resonator, but a \emph{well-centred} sharp resonator still wins where it is
accurate (MSE $\approx 0.22\text{--}0.33$ within $\pm15\%$). \emph{Nonstationary}:
under chirp the adaptive controller wins ($0.40$ vs.\ $0.41$ oracle, $0.51$ PI);
under regime switching PI ($0.59$) edges out adaptive ($0.79$) because of
re-identification transients---reported honestly. \emph{Supports
Prop.~\ref{prop:adapt} and its failure mode.}

\begin{figure}[!t]
\centering
\includegraphics[width=\columnwidth]{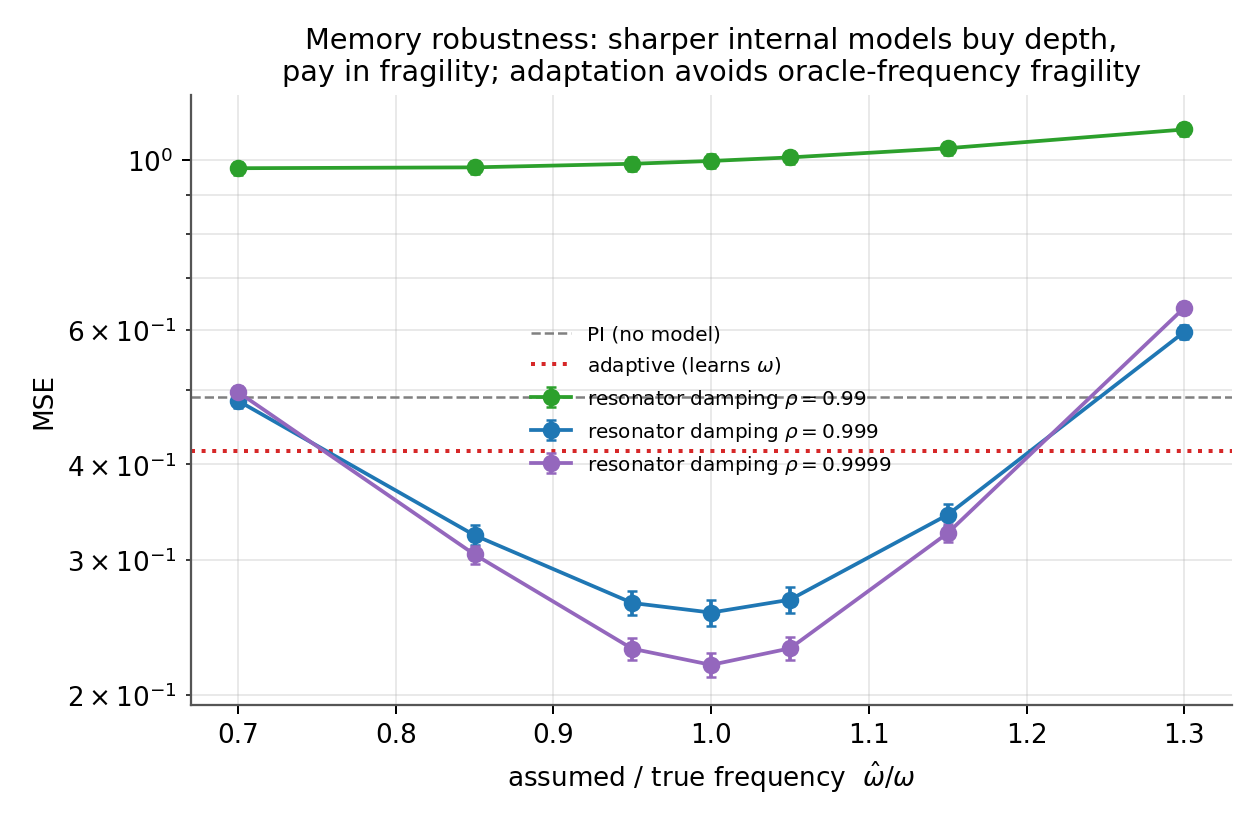}
\caption{Block B: frequency mismatch $\times$ resonator damping. Sharper internal
models buy suppression depth but pay in fragility; adaptation re-centres the
notch and escapes the strongly-mistuned regime.}
\label{fig:mismatch}
\end{figure}

\begin{figure*}[!t]
\centering
\includegraphics[width=0.9\textwidth]{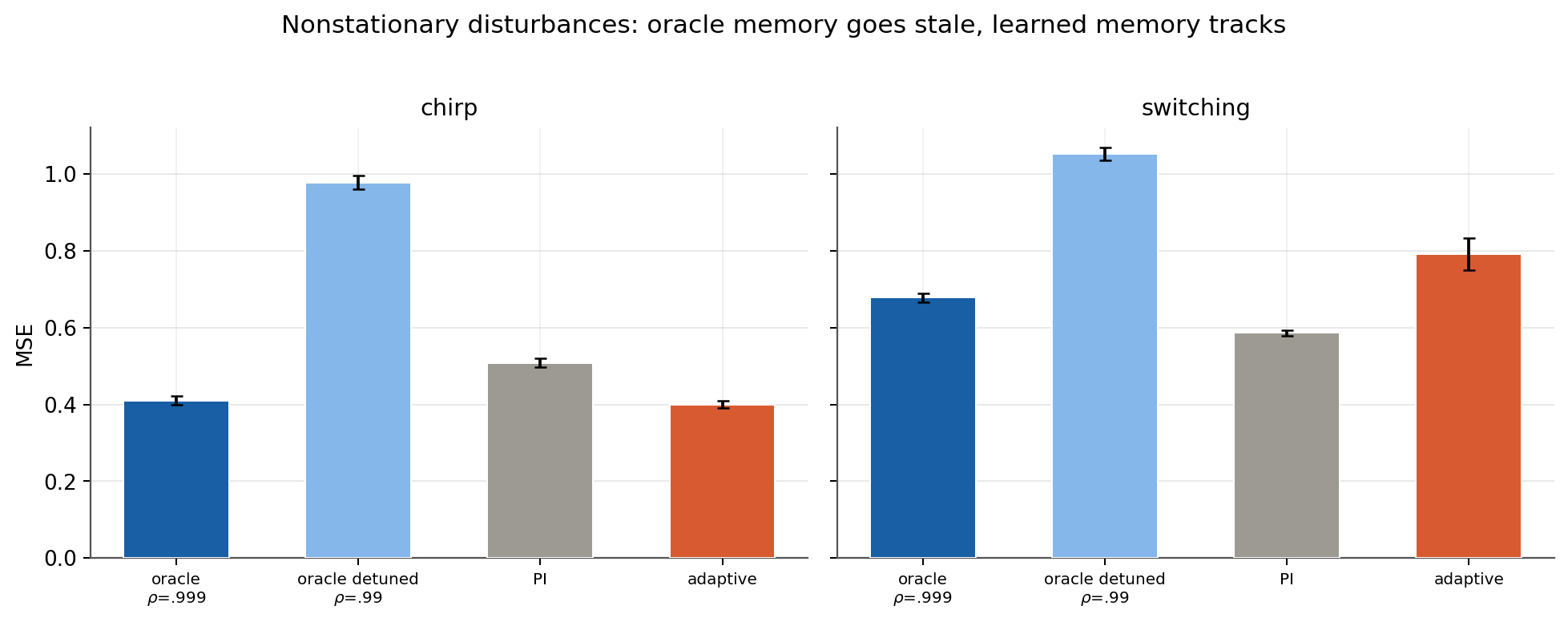}
\caption{Block B, nonstationary disturbances. Under a chirp, learned memory
tracks; under regime switching, a robust memoryless loop (PI) can beat the
adaptive controller because each switch triggers a re-identification transient.}
\label{fig:nonstat}
\end{figure*}

\begin{figure}[!t]
\centering
\includegraphics[width=\columnwidth]{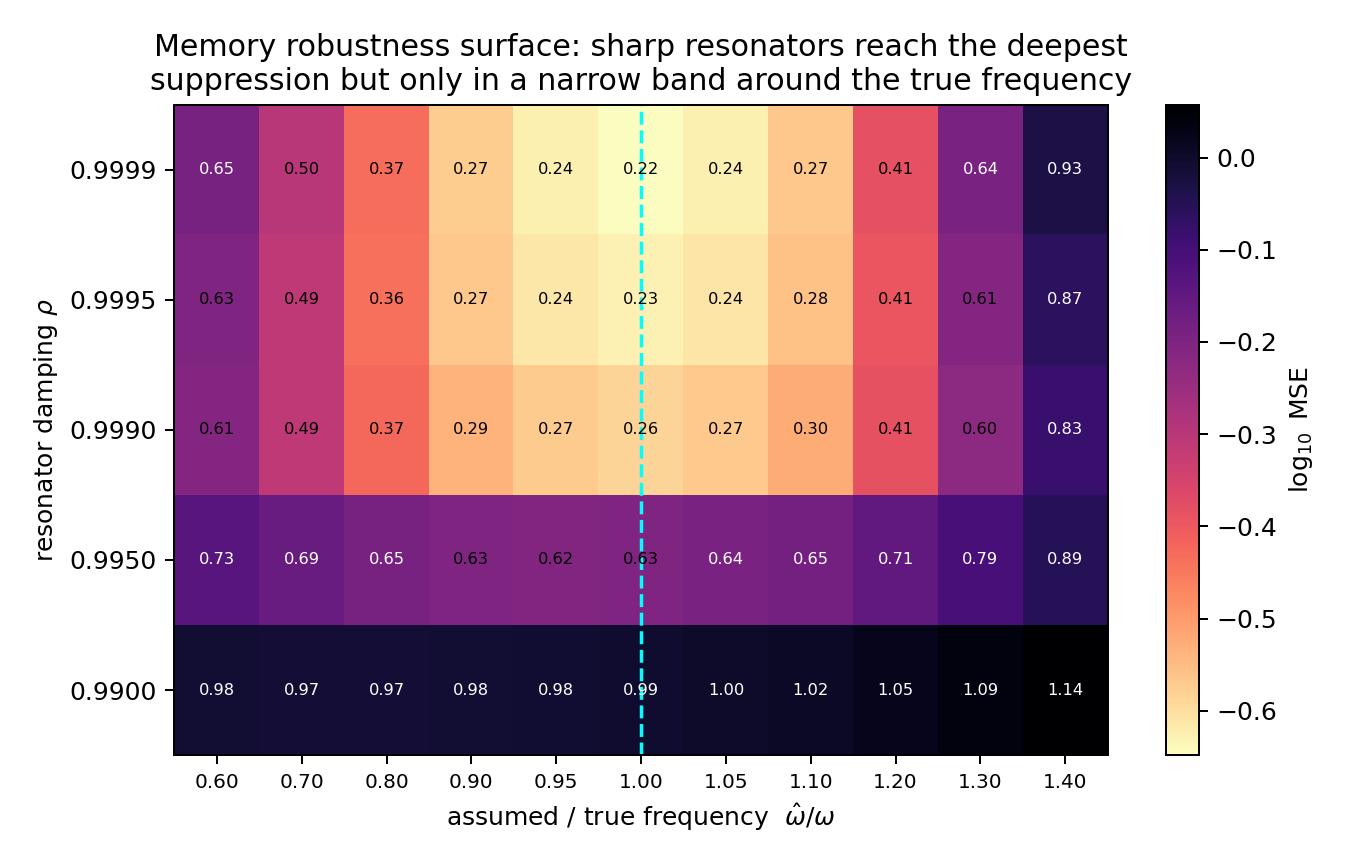}
\caption{Memory robustness surface (Block B). Color: $\log_{10}$ MSE over the
two internal-model parameters, assumed/true frequency $\hat\omega/\omega$ and
resonator damping $\rho$ (annotated numbers are the raw MSE). Sharp resonators
(top rows) reach the deepest suppression but only in a narrow band around the
true frequency; damped resonators (bottom) are robust but shallow---the
depth/robustness waterbed made quantitative.}
\label{fig:robust2d}
\end{figure}

\textbf{Block C --- the delay theorem (Fig.~\ref{fig:delay}).} Across 15
$(a,\tau)$ combinations, the simulated optimal MSE matches the semi-analytic
Riccati floor of Prop.~\ref{prop:delay} to $<1\%$; PI diverges past
$\tau\approx 10\text{--}15$. The floor grows with $\tau$, the noise memory $a$,
and the sensing noise. \emph{Verifies Prop.~\ref{prop:delay}.}

\begin{figure}[!t]
\centering
\includegraphics[width=\columnwidth]{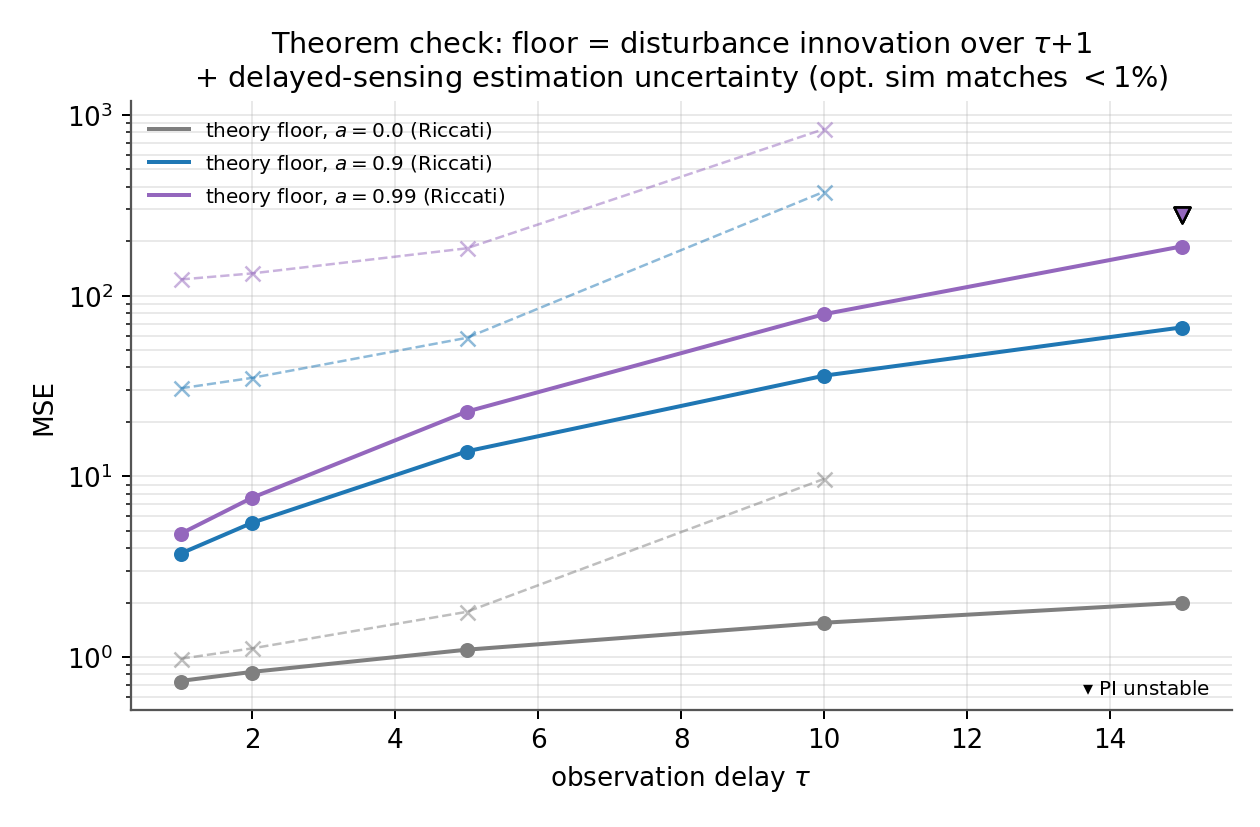}
\caption{Block C: the optimal controller's simulated MSE (markers) matches the
semi-analytic Riccati floor (lines) to $<1\%$ (max $0.89\%$) across delays
$\tau$ and noise memories $a$; the floor is the disturbance innovation over
$\tau{+}1$ plus the delayed-sensing estimation uncertainty. PI (dashed) is shown
where stable and marked ($\blacktriangledown$) where it diverges.}
\label{fig:delay}
\end{figure}

\textbf{Block E --- a mild-nonlinearity robustness check
(Fig.~\ref{fig:nonlinear}).} As a preliminary check---not a claim about general
nonlinear plants---we add a bounded term $0.15\sin x$ (this is neither chaotic
nor non-minimum-phase). The memory ordering holds \emph{once width is
sufficient} ($N{=}300$: $d{=}2$ $1.76 < d{=}1$ $2.32 < d{=}0$ $3.65$), but at
small $N{=}10$ integrative memory can amplify measurement noise---PI ($d{=}1$,
$16.4$) is then worse than plain P ($d{=}0$, $13.6$), while the resonator
($d{=}2$, $6.5$) still leads. Under nonlinear+switching the stale oracle
resonator ($2.07$) and adaptive ($2.45$) both beat PI ($3.75$). The linear
resource accounting is thus robust to this class of \emph{bounded}
nonlinearities; generalization to unbounded nonlinearities (e.g.\ $x^3$, or
non-minimum-phase plants) is left for future work.

\textbf{Block F --- spatial generalization (Fig.~\ref{fig:spatial}).} To check
that the accounting is not an artifact of the scalar testbed, we replace the
single object with a ring of $12$ diffusively-coupled integrator nodes driven by
a \emph{shared slow mode} plus local per-node noise; each node runs a local
controller fed by $N$ agents' delayed, noisy observations. Once the population is
large enough to average the local noise, the same ordering emerges: the
memoryless controller (P) saturates at a floor ($\approx 3.7$)
independent of $N$ because the shared structured mode survives averaging, the
integrator (PI) does better, and a resonator tuned to the shared mode reaches
the lowest floor ($0.23$ at $N{=}300$). Width cures the local noise; only memory
cures the shared structure---exactly as in the scalar case. This further holds on
a two-dimensional $6\times6$ grid and, more strikingly, under \emph{partial
observability} (Fig.~\ref{fig:spatial2d}): when only half the nodes carry
sensors and actuators, width alone (P) is stuck at a large floor ($\approx18.5$),
but a resonator on the shared mode still attains the lowest error---the memory
ordering survives both the coupling dimension and incomplete sensing.

\begin{figure*}[!t]
\centering
\includegraphics[width=0.9\textwidth]{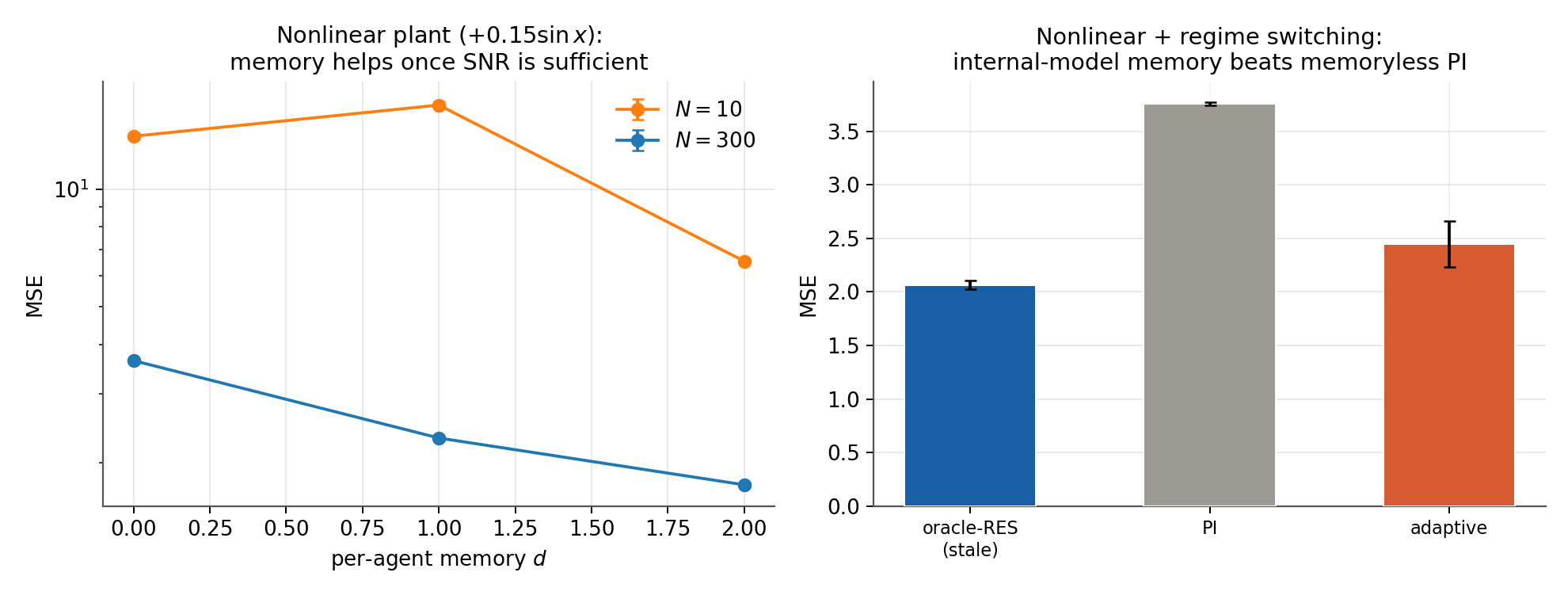}
\caption{Block E: on a nonlinear plant, the memory ordering persists once width
is sufficient (left); under nonlinear plus regime switching, internal-model
memory (a stale oracle resonator and the adaptive controller) beats memoryless
PI (right). The oracle resonator still edges the adaptive controller here.}
\label{fig:nonlinear}
\end{figure*}

\begin{figure}[!t]
\centering
\includegraphics[width=\columnwidth]{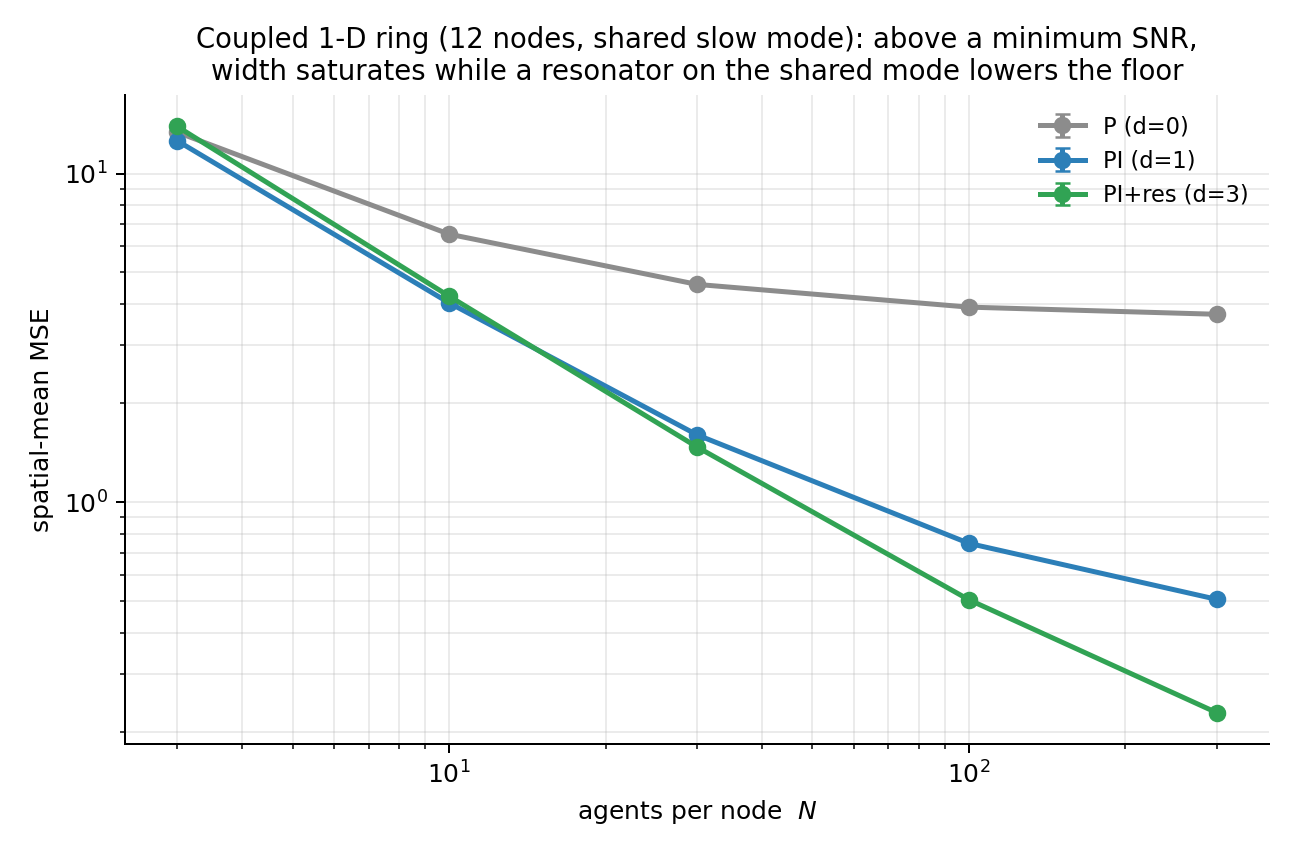}
\caption{Block F: a ring of 12 diffusively-coupled nodes with a shared slow
disturbance mode. Width (P, gray) saturates at an $N$-independent floor because
the shared mode survives averaging; an integrator (PI) improves; a resonator on
the shared mode (green) reaches the lowest floor once the population is large
enough to average the local noise (the ordering emerges above a minimum SNR).
The scalar-case accounting carries over to the spatially-extended plant.}
\label{fig:spatial}
\end{figure}

\begin{figure*}[!t]
\centering
\includegraphics[width=0.92\textwidth]{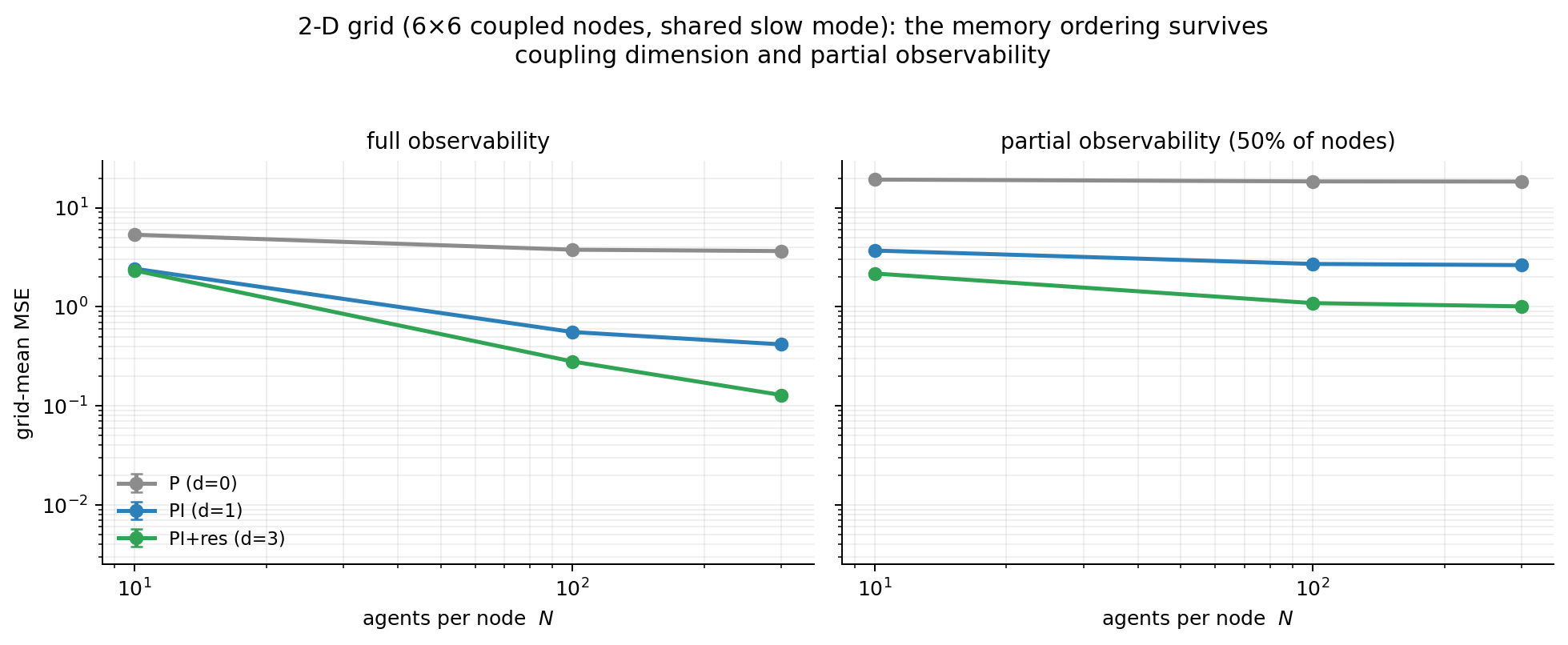}
\caption{Block F, continued: a two-dimensional $6\times6$ coupled grid. Left:
full observability. Right: partial observability, where only half the nodes
carry sensors/actuators and the rest are steered through diffusion. In both
cases the memory ordering (P $>$ PI $>$ PI+resonator) is preserved; partial
sensing raises every floor but does not change which resource lowers it.}
\label{fig:spatial2d}
\end{figure*}

\section{The Resource Triangle and Design Rules}
\label{sec:triangle}
Table~\ref{tab:triangle} summarizes the accounting, and Fig.~\ref{fig:zones}
reads it off the width$\times$memory map as three regimes: a variance-limited
corner cured by width, a memory-limited corner cured by internal-model states,
and the delay/innovation floor that neither removes.

\begin{figure}[!t]
\centering
\includegraphics[width=\columnwidth]{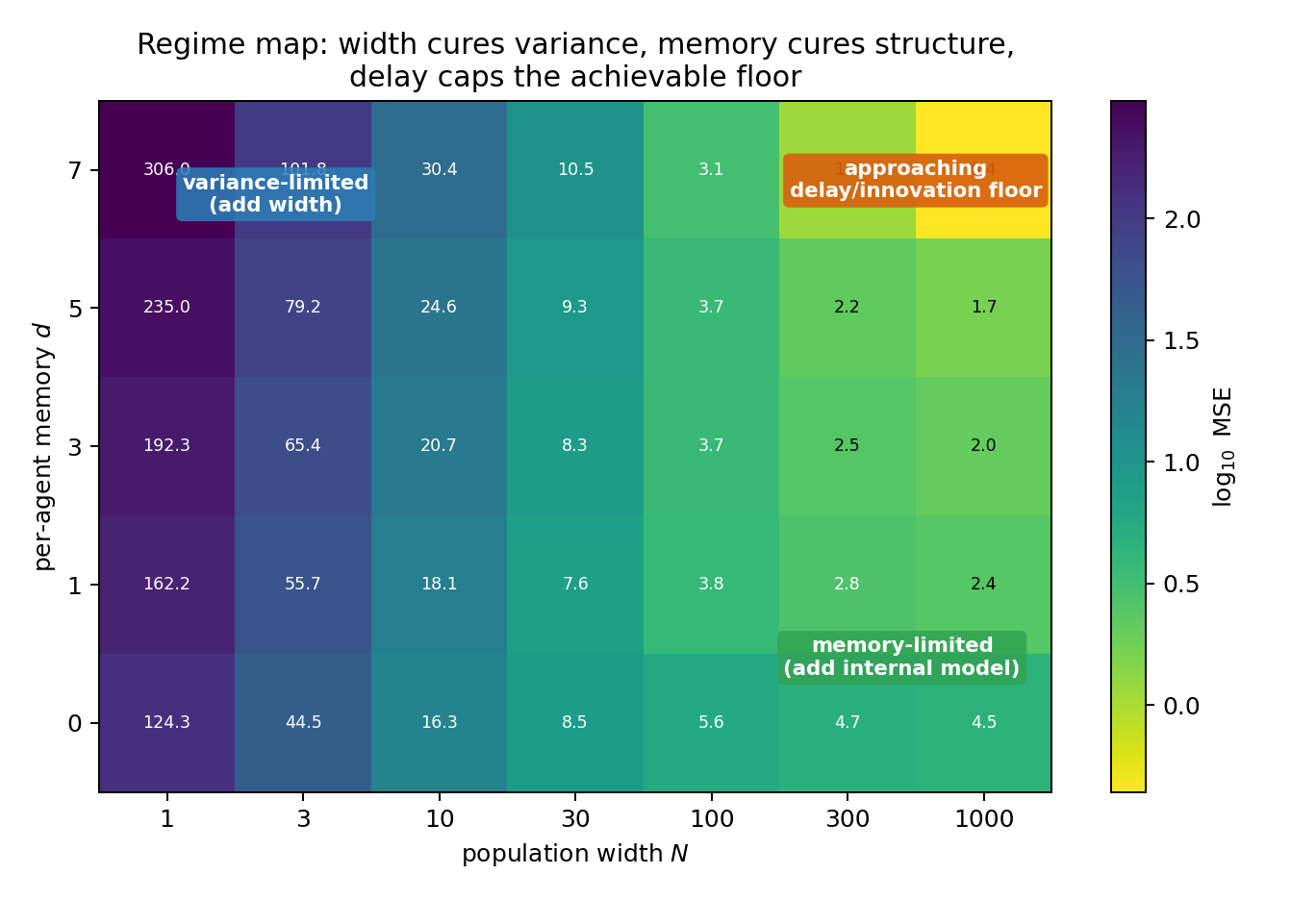}
\caption{Regime map, colored by $\log_{10}$ MSE (annotated numbers are the raw
MSE). At low $N$ the system is variance-limited (add width); at high $N$, low
$d$ it is memory-limited (add an internal-model state); the bottom-right corner
approaches the delay/innovation floor that neither resource can lower.}
\label{fig:zones}
\end{figure}

\begin{table*}[!t]
\caption{The resource triangle: what each resource buys, what it cannot, and the governing law.}
\label{tab:triangle}
\centering
\renewcommand{\arraystretch}{1.3}
\begin{tabular}{@{}p{2.6cm} p{3.6cm} p{3.6cm} p{5.2cm}@{}}
\toprule
\textbf{Resource} & \textbf{Buys} & \textbf{Does \emph{not} buy} & \textbf{Rate / law}\\
\midrule
Width $N$ & averaging of i.i.d.\ noise & structural/spectral content &
$\propto 1/N$ down to a floor; saturates at $N^\ast$\\
Per-agent memory $d$ & spectral coverage, internal model & delay compensation by
itself & need $d \ge 2\cdot\#\text{bands}$; floor set by uncovered power (Prop.~\ref{prop:width})\\
Prediction over $\tau$ & cancels the \emph{predictable} part of the disturbance (shifts the floor curve down) & the \emph{innovation}: the unpredictable noise over the horizon is never removed & floor $=$ disturbance innovation over $\tau{+}1$ \emph{plus} delayed-sensing estimation uncertainty (Prop.~\ref{prop:delay})\\
Adaptation & replaces oracle knowledge of $\omega$ & instant tracking &
$\approx$ oracle $-$ small gap; loses under fast switching (Prop.~\ref{prop:adapt})\\
Resonator damping & width of robustness & depth at that width & notch waterbed:
depth $\times$ bandwidth $\approx$ const\\
\bottomrule
\end{tabular}
\end{table*}

\textbf{Golden rules.}
\begin{enumerate}
\item \emph{Count the disturbance modes first.} Each persistent band needs its
own internal-model state ($\sim 2$ states/band). No number of agents substitutes
for a missing mode.
\item \emph{Add width only up to $N^\ast$.} Beyond the point where $1/N$
reduction meets the structural floor, agents are ballast; $N^\ast$ is computable
from $\sigma_m^2$ and the target floor.
\item \emph{Latency is the hard wall.} The floor equals the environment's
unpredictability over the observation (feedback) delay. Buy a predictor or cut latency, not
more agents.
\item \emph{Under model uncertainty, detune or adapt.} Trade notch depth for
robustness width; go adaptive only if regimes outlive the identification
transient.
\end{enumerate}

\section{Limitations and Outlook}
\label{sec:limits}
The core testbed is a single scalar object observed by $N$ homogeneous agents,
which isolates the width/memory/delay axes; Block~F confirms the accounting
survives on a spatially-extended chain and a two-dimensional grid, including
under 50\% partial observability (Figs.~\ref{fig:spatial},~\ref{fig:spatial2d}).
Further, the plant is a 1-D integrator per node (a nonlinear variant is tested,
but genuinely chaotic, higher-order, or non-minimum-phase plants are not); the
tuning grid is coarse (the \emph{ordering} of floors is robust, though constants
may shift); the adaptive controller assumes a few narrow tones (broadband or
heavy-tailed disturbances are out of scope); it loses to a robust PI under fast
switching (reported, not tuned away). The hierarchical baseline is a single
tuned two-loop cascade; Proposition~\ref{prop:depth} is a counterexample to a
universal claim and does not rule out a better-designed hierarchy (e.g.\
$H_\infty$ or model-based nested predictors) matching the flat swarm. We also
assume no inter-agent communication; Section~\ref{sec:setup} discusses how
distributed/consensus estimation could partially close the gap to the
centralized frontier. Measurement noise is i.i.d.\ across agents and actuators
are unconstrained; correlated or heterogeneous sensing and actuator saturation
are not modeled. We do not claim to prove impossibility results; we exhibit a
counterexample to a strong claim and propose a sharper, classically-grounded
resource model, verified empirically. Each of these---a second hierarchical
design, distributed estimation, richer plant classes, correlated/heterogeneous
noise, and input constraints---is a concrete, testable extension of the same
resource accounting rather than a departure from it.

\section{Discussion}
\label{sec:discussion}
\textbf{A trilemma, not a dichotomy.} The debate the target preprint joins is
usually posed as flat \emph{versus} deep. Our results recast it as a triangle of
three non-interchangeable resources (Fig.~\ref{fig:concept}): width buys
averaging, internal-model memory buys structure, and delay imposes a floor.
Unlike the blockchain ``scalability trilemma'', where one must sacrifice one of
three desirable properties, here the three resources are not competing goods but
\emph{different jobs}: the lesson is that no amount of one substitutes for
another, and that the ``depth'' the pessimistic view demands is bought by
recurrent internal-model memory rather than by architectural nesting.

\textbf{A biological reading.} The account has a suggestive parallel in social
insects. A colony does not build a hierarchy of ``generals'' to steer each ant;
hierarchy is expensive and slow. Instead evolution hard-wires into each
individual an internal model of the world---instinct---while the size of the
swarm (width) averages out individual sensing noise. On our axes, instinct is
per-agent internal-model memory, colony size is width, and the collective
thrives precisely because structured regularities of the environment are encoded
in every agent rather than computed by a central authority. The same account
predicts the failure mode: when the environment changes faster than instincts
can be updated---faster than the identification time in our Block~B switching
experiment---the swarm cannot re-tune every individual quickly, and a rigid
internal model is worse than a simple robust reflex. Two caveats keep the analogy
honest: many insects retain limited online plasticity (bees learn), and
evolution itself is a slow \emph{outer} adaptation loop on a generational
timescale---so the colony does adapt, but on a separated, much slower timescale,
which is itself an instance of the depth-as-timescale-separation idea rather than
a counterexample to it. We offer this as intuition, not mechanism.

\textbf{Implications for LLM-agent systems.} With the caveat that our testbed is
a physical control system, the accounting offers a hypothesis for cognitive
collectives. If ``memory'' maps to an agent's context and retrieval (its internal
model of the task domain), ``width'' to the number of parallel agents, and
``delay'' to reasoning or communication latency, then the recurring empirical
findings---that a strong single agent with adequate memory can match a
multi-agent workflow \cite{xu2026single,tran2026single}, and that scaling helps
mainly when the added agents bring \emph{diversity} \cite{qian2024scaling}---are
exactly what the trilemma predicts: replication (width) averages noise but cannot
add missing domain structure, while diversity widens coverage. Whether the
quantitative laws transfer is an open question our released harness is meant to
probe.

\section{Conclusion}
Recasting the flat/deep question as resource accounting resolves the apparent
paradox in both control swarms and LLM-agent collectives: width buys only the
averageable part of the uncertainty, internal-model memory buys the structured
part, and delay imposes an innovation floor that neither can remove.
``Hierarchy'' is one way to organize predictive memory across timescales, not the
only one---which is why a strong \emph{single} agent with adequate memory can
match a multi-agent workflow \cite{xu2026single,tran2026single}, and why scaling
helps mainly when it adds \emph{diversity} that widens spectral coverage
\cite{qian2024scaling}. Whether the same accounting quantitatively governs
language-model agents---where ``memory'' is context and retrieval and ``delay''
is reasoning latency---is a concrete and, we believe, testable next question; our
released harness (Appendix~\ref{app:explore}) is a first step.

\appendices
%\clearpage
\section{Exploratory Studies}
\label{app:explore}
These studies are exploratory and not part of the main evidential line.

\textbf{A. Emergent timescale specialization (E6).} Flat collectives of
learnable leaky integrators trained by an evolution strategy \cite{salimans2017es}
reliably split off a \emph{single} slow agent (learned timescales $20\text{--}1500$
vs.\ $1\text{--}8$ for the rest); the heterogeneous multi-band condition is more
stable ($0/5$ unstable runs) than the homogeneous one ($2/5$). This suggests the
slow state that Prop.~\ref{prop:depth} inserts by hand can instead \emph{emerge}
in a learned flat collective---heterogeneity creating a form of virtual temporal
depth. We stress the exploratory status: populations are tiny (8 agents), the ES
is simple, and a slow specialist appears (less consistently) even in the
single-band control, so the evidence supports the heterogeneous-vs-homogeneous
contrast, not a sharp multi-band claim.

\textbf{B. LLM-agent analogue (E7, harness only).} We provide a harness casting
flat, flat-ensemble, and ``deep'' (fast-reactor plus slow-strategist) LLM agents
onto the same regulation task, with a mock backend for pipeline validation and an
API backend for real runs. The mock backend yields \emph{no} scientific result;
the harness is offered as a future benchmark for whether the width/memory/delay
accounting transfers to language-model agents.

\begin{figure}[!t]
\centering
\includegraphics[width=\columnwidth]{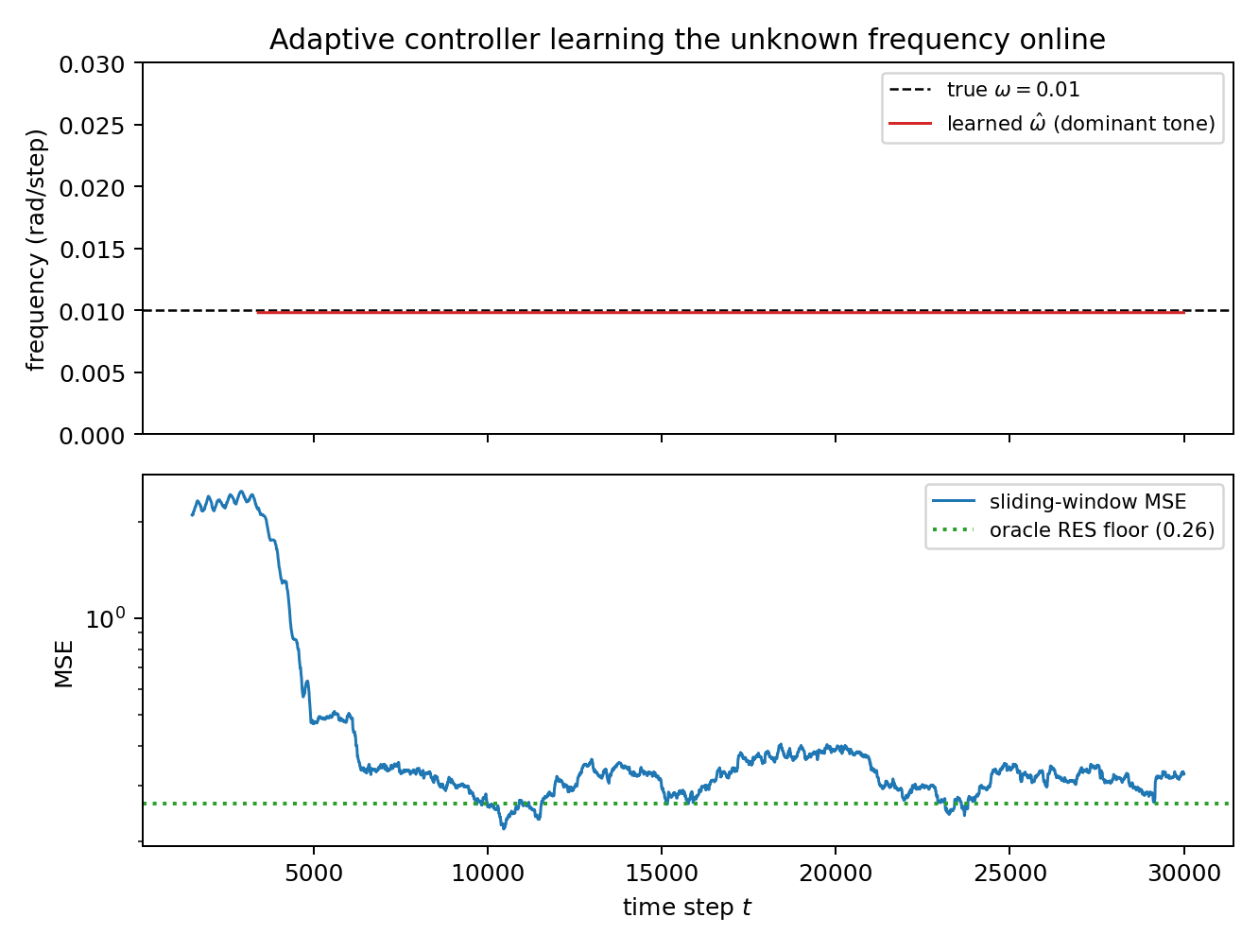}
\caption{Adaptive controller internals (no oracle). Top: the online frequency
estimate locks onto the true unknown $\omega$ within a few thousand steps.
Bottom: the sliding-window MSE falls toward the oracle resonator floor as the
model converges---illustrating the identification threshold and convergence gap
of Proposition~\ref{prop:adapt}.}
\label{fig:adaptivedyn}
\end{figure}

\textbf{C. Adaptive controller dynamics (Fig.~\ref{fig:adaptivedyn}).} The
no-oracle adaptive controller is not a black box: its periodogram-driven
frequency estimate converges to the true $\omega$ online, and its
sliding-window error approaches the oracle floor once the estimate locks,
making concrete the threshold and gap of Proposition~\ref{prop:adapt}.

\section*{Data and Code Availability}
All code, configurations, seeds, and raw results needed to reproduce every
figure and number in this paper are released under the MIT license at
\url{https://github.com/KuznetsovKarazin/causal-depth-limits}, with a
versioned, citable snapshot archived on Zenodo (DOI:
\href{https://doi.org/10.5281/zenodo.21295164}{10.5281/zenodo.21295164}). Results
regenerate deterministically from the released seeds; a Dockerfile, a Conda
environment, and continuous integration (\texttt{make smoke},
\texttt{make check}) are included. No proprietary data or paid API access is
required for the core results; the optional real-LLM harness (Appendix,
\texttt{make e7-real}) needs an Anthropic API key and its output is excluded
from the paper's claims.

\bibliographystyle{IEEEtran}
\bibliography{refs}

\begin{IEEEbiography}[{\includegraphics[width=1in,height=1.25in,clip,keepaspectratio]{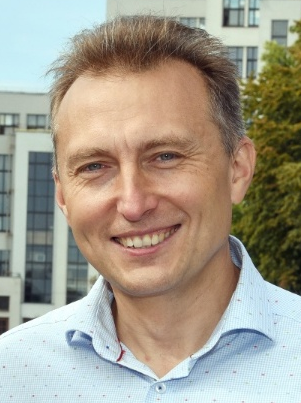}}]{Oleksandr Kuznetsov}
(Member, IEEE) is an Associate Professor at the Department of Theoretical and
Applied Sciences, eCampus University, Italy, and a Full Professor at the
Department of Intelligent Software Systems and Technologies, V.N. Karazin Kharkiv
National University, Ukraine.
His distinguished career spans over two decades with leadership roles across
academia, industry, and defense sectors.
Professor Kuznetsov's research expertise encompasses applied cryptography,
post-quantum security, blockchain and decentralized systems, artificial
intelligence in cybersecurity, biometric authentication, steganography, and IoT
security.
\end{IEEEbiography}

\begin{IEEEbiography}[{\includegraphics[width=1in,height=1.25in,clip,keepaspectratio]{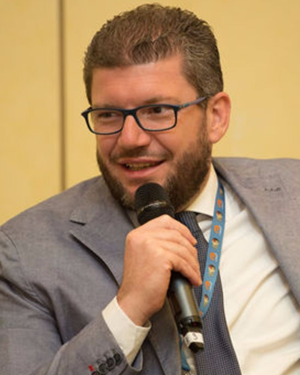}}]{Emanuele Frontoni}
(Member, IEEE) is a Full Professor of computer science with the University of
Macerata and the Co-Director of the VRAI Vision Robotics \& Artificial
Intelligence Lab.
His research interests include computer vision and artificial intelligence with
applications in robotics, video analysis, human behavior analysis, extended
reality and digital humanities.
He is the author of over 230 international articles and collaborates with numerous
national and international companies in technology transfer and innovation
activities.
He has been program chair or general chair of various international conferences
and summer schools (e.g.\ IEEE/ASME MESA 2016 and 2017, IEEE ECMR 2017, BigDat
2020, DeepLearn 2021) and co-organizer of many international workshops (e.g.\
DeepRetail@ICPR 2020, D2CH@CVPR 2021, AI4DH@ICIAP 2022).
\end{IEEEbiography}

\end{document}